\documentclass[english]{article}
\usepackage[T1]{fontenc}
\usepackage[utf8]{inputenc}
\usepackage{geometry}
\usepackage{amsthm}
\usepackage{amsmath}

\makeatletter
\theoremstyle{plain}
\newtheorem{thm}{\protect\theoremname}
  \theoremstyle{definition}
  \newtheorem{defn}[thm]{\protect\definitionname}
  \theoremstyle{plain}
  \newtheorem{lem}[thm]{\protect\lemmaname}
  \theoremstyle{plain}
  \newtheorem*{lem*}{\protect\lemmaname}

\usepackage[auth-lg]{authblk}

   \author{Rotem Arnon-Friedman}
   \author{Amnon Ta-Shma}
   \affil{The Blavatnik School of Computer Science, Tel-Aviv University, Israel}

\usepackage{tikz}
\usepackage{pgfplots}
\usepackage{pgf}
\usetikzlibrary{patterns}
\usepackage{float}

\usepackage{graphicx}

\usepackage{multicol}
\usepackage{appendix}

\usepackage{babel}
  \providecommand{\definitionname}{Definition}
  \providecommand{\lemmaname}{Lemma}
\providecommand{\theoremname}{Theorem}

\begin{document}

\title{Limits of privacy amplification against non-signalling memory attacks}
\maketitle

\makeatother

\begin{abstract}
	The task of privacy amplification, in which Alice holds some partially secret information with respect to an adversary Eve and wishes to distill it until it is completely secret, is known to be solvable almost optimally in both the classical and quantum worlds. Unfortunately, when considering an adversary who is limited only by non-signalling constraints such a statement cannot be made in general. We here consider systems which violate the chained Bell inequality and prove that under the natural assumptions of a time-ordered non-signalling system, which allow past subsystems to signal future subsystems (using the device's memory for example), super-polynomial privacy amplification by any hashing is impossible. This is of great relevance when considering practical device independent key distribution protocols which assume a super-quantum adversary. 
\end{abstract}
\bigskip{}

\begin{multicols}{2}
\section{Introduction}

\subsubsection*{Device independent key distribution}

Key distribution is the task of creating a shared secret string, called the key, between two parties. In contrast to classical key distribution protocols, which base their security on the computational power of the adversary, quantum key distribution (QKD) protocols are resilient against quantum adversaries with unbounded computational power. However, in order to apply traditional QKD security proofs, such as security proofs for the BB84 protocol \cite{BB84}, one should be able to fully characterise the devices on which the protocol is being executed. Failing to do so can introduce security flaws which can be exploited by the adversary \cite{gerhardt2011full}. Unfortunately, giving a full characterisation of quantum devices is usually an impractical task. 

Due to this difficulty, in the past few years there has been a growing interest in device independent QKD (DIQKD). In DIQKD protocols \cite{mayers1998quantum,pironio2009device} we assume that the system on which the protocol is being executed was made and given to the honest parties Alice and Bob by a malicious adversary Eve. We therefore ought to consider the system, which we know nothing about, as a black box, and the security proof cannot be based on the internal functioning of the device. 

How can this be done? As was first shown in \cite{acin2006bell}, security proofs for DIQKD can be based on observed non-local correlations between Alice and Bob, i.e., on the correlations of the outputs they get from their systems. If the correlations they observe violate some Bell inequality, such as the CHSH inequality \cite{CHSH} or other more general chained Bell inequities \cite{braunstein1990wringing,Barrett2006Maximally}, and if Alice and Bob enforce a non-signalling condition between them in order to make sure that these correlations are indeed non-local, then they can be sure that some secrecy is available to them \cite{Barrett2006Maximally}.

The first DIQKD protocol which was proven secure was a protocol by Barret, Hardy and Kent (BHK) \cite{barrett2005no}. Although this protocol cannot tolerate a reasonable amount of noise, it showed that the task of DIQKD is in principle possible. Moreover, the BHK protocol security proof applies not only against quantum adversaries, but also against non-signalling adversaries. 

When considering a non-signalling adversary the only thing which limits the adversary is the non-signalling principle. That is, the adversary has super-quantum power; however, if Alice and Bob enforce some local non-signalling constraints then these cannot be broken by the adversary. Such constraints can be enforced by shielding and isolating the devices or by placing them in a space-like separated way. For example, if Alice and Bob perform their measurements in a space-like separated way, then according to relativity theory, Alice cannot use her system in order to signal Bob and vice-versa. 

After the BHK protocol, several other DIQKD protocols, such as \cite{hanggi2009quantum,masanes2009universally}, have been proven secure, but all using an impractical assumption; in order to guarantee security each subsystem used in the protocol must be isolated from all other subsystem, such that they cannot signal each other. For example, if Alice gets $n$ systems from Eve, each producing one bit, she must isolate each of these systems, in order to make sure that no information leaks from one system to another. Such an harsh constraint, which we call the full non-signalling constraint, eliminates the possibility of devices with memory. 

Recently a new protocol, which does not share this drawback, was proven secure \cite{barrett2012unconditionally}. The sole assumption about the non-signalling constraints of the system in this protocol is that Alice, Bob and Eve cannot signal each other using the system, which is a minimal requirement from any cryptographic protocol%
\footnote{If Alice's system can signal Eve's system then Alice's secret can leak to Eve completely. If Alice's system can signal Bob's system, then the correlations they observe are not necessarily non-local and could have been produced by a deterministic system. This implies that Eve can get all the information that Alice and Bob have as well. %
}. However, this protocol, like the BHK protocol, cannot tolerate any reasonable amount of noise.

\subsubsection*{Privacy amplification}

In this paper we consider a simpler problem, called privacy amplification (PA). In the PA problem Alice holds some information which is only partially secret with respect to an adversary, Eve. Alice's goal is to distill her information, to a shorter string, which is completely (or almost completely) secret. Note that in the PA problem we only want Alice to have a secret key with respect to the adversary, while in QKD we also want Bob to hold the same key as Alice. Therefore PA is easier than QKD. 

In order to understand what exactly is the PA problem, consider the following scenario. Assume that Alice has a system, a black box, which produces for her a partially secret bit or a string, $X$. By saying that $X$ is partially secret we mean that there is some entropy in $X$ conditioned on Eve's knowledge about $X$. One would hope that by letting Alice use several such systems, which will produce several partially secret bits $X_{1}$, $X_{2}$..., $X_{n}$, she will have enough entropy in order to produce a more secret bit or a string, $K$, out of them, or in other words, she will be able to amplify the privacy of her key. The idea behind the PA protocols is to apply some hash function%
\footnote{The hash function might also take a random seed of size $m$ as an additional input; in that case $f:\{0,1\}^{n}\times\{0,1\}^{m}\rightarrow\{0,1\}^{|K|}$.%
} $f:\{0,1\}^{n}\rightarrow\{0,1\}^{|K|}$ (for $|K|<n$) to $X_{1}$, $X_{2}$..., $X_{n}$ in order to receive a shorter, but more secret, bit string $K$. The amount of secrecy is usually measured by the distance of the actual system of Alice and Eve from an ideal system, in which $K$ is uniformly distributed and uncorrelated to Eve's system. This will be defined formally in the following section. 

Since QKD in the presence of a non-signalling adversary is possible if we assume that Alice's and Bob's systems fulfil the full non-signalling conditions  \cite{hanggi2009quantum,masanes2009universally}, PA is also possible in this setting. However, it was already proven in \cite{hanggi2010impossibility} that PA is impossible if we impose non-signalling conditions only between Alice and Bob,\footnote{In contrast to the QKD problem, when considering the PA problem the only goal of Bob is to establish non-local correlations with Alice.}, i.e., Alice and Bob cannot signal each other, while signalling within their systems is possible. Recently, the impossibility result of \cite{hanggi2010impossibility} was extended to an even more general case \cite{arnon2012towards}.

A more realistic assumption to consider is that in addition to the non-signalling assumption between Alice and Bob, within the system of the parties signalling is possible only from the past to the future and not the other way around. These are natural assumptions when considering a protocol in which Alice and Bob each use just one device with memory. In that case, the inputs and outputs of past measurements (which were saved in the memory of the device) can affect the outputs of future measurements. Such conditions, which we call time-ordered non-signalling conditions, are defined formally in Definition \ref{time-orderd-non-signalling}. 

In contrast to the full non-signalling conditions, the time-ordered non-signalling conditions are easy to ensure. Alice and Bob can both shield their entire system (as has to be done anyhow in order to make sure that no information leaks straight to the adversary) and therefore signalling will be impossible between them. Moreover, when running the protocol, they will perform their measurements in a sequential manner; the first system will be measured in the beginning, then the second one and so on. This will make sure (as long as we believe that messages cannot be sent from the future to the past) that signalling is possible only in the forward direction of time. In fact, these are the non-signalling conditions that one ``gets for free'' when performing an experiment of QKD. For example, an entanglement-based protocol in which Alice and Bob receive entangled photons and measure them one after another using the same device will lead to the time-ordered non-signalling conditions. If Alice's and Bob's devices have memory then information from past measurements can be available for future measurements, i.e., signalling is possible from the past to the future but not the other way around.

In this paper we ask the following question. Under the assumptions of time-ordered non-signalling system, is privacy amplification against non-signalling adversaries possible? We give an example for a system which fulfils all the time-ordered non-signalling conditions, and in which super-polynomial PA is impossible. More precisely, we prove that for protocols which are based on a violation of chained Bell inequalities, under the assumption of a time-ordered non-signalling system, super-polynomial PA is impossible by any hashing. That is, when using $n$ black boxes, each producing a partially secret bit, the adversary can always get a great amount of information about the hashing result; at least as high as $\Omega\left(\frac{1}{n}\right)$. 

Although this proves that super-polynomial PA is impossible under these conditions, this is still a partial answer to our question for two reasons. First, there might still be some other system, in which the secrecy is based on a different Bell inequality, for which exponential PA is possible. Second, in this paper we show that, independently of which hash function Alice is using, Eve can bias the key by at least $\Omega\left(\frac{1}{n}\right)$; but can we find a specific hash function for which she cannot do any better than this? That is, is this result tight? Therefore, the question of whether (linear) privacy amplification is at all possible remains open. 

\section{Preliminaries}

\subsubsection*{Chained Bell inequalities}

For two correlated random variables $X,U$ we denote the conditional probability distribution of $X$ given $U$ by $P_{X|U}\left(x|u\right)=\Pr\left(X=x|U=u\right)$. A bipartite system is defined by the joint input-output distribution $P_{XY|UV}$, where $U$ and $X$ are usually Alice's input and output respectively, while $V$ and $Y$ are Bob's input and output. When considering a tripartite system which includes Eve, $P_{XYZ|UVW}$, Eve's input and output are $W$ and $Z$.

Bell proved that entangled quantum states can display non-local correlations under measurements \cite{bell64}. We consider the following Bell-type experiments. Alice and Bob share a bipartite system $P_{XY|UV}$ where $U\in\{0,2,...,2N-2\}$ and $V\in\{1,3,...,2N-1\}$. We define a set of allowed input paris for Alice and Bob to be $G_N = \{(u,v)|u\in U, v\in V, |u-v|=1\}\bigcup\{(0,2N-1)\}$. For each measurement of Alice $U$, and each measurement of Bob $V$, there are two possible outcomes, 0 or 1. That is, $X,Y\in\{0,1\}$. The relevant Bell inequality then reads \cite{braunstein1990wringing,Barrett2006Maximally} 
\begin{multline}
	I_{N}=P\left(X=Y|U=0,V=2N-1\right)+\\
	\underset{\begin{array}{c}
	u,v\\
	|u-v|=1
	\end{array}}{\sum}P\left(X\neq Y|U=u,V=v\right)\geq1.\label{eq:chained-inequality}
\end{multline}
This implies that correlations which satisfy $I_{N}<1$ are non-local and cannot be described by shared randomness of the parties. For $N=2$ this inequality is the CHSH inequality \cite{CHSH}.

For the maximally entangled state $|\Phi+\rangle=\frac{1}{\sqrt{2}}\left(|00\rangle+|11\rangle\right)$, if Alice's measurements are in the basis $\left\{ \mathrm{cos} \frac{\theta}{2}|0\rangle+\mathrm{sin}\frac{\theta}{2}|1\rangle,\mathrm{sin}\frac{\theta}{2}|0\rangle-\mathrm{cos}\frac{\theta}{2}|1\rangle\right\} $ for $\theta=\frac{\pi U}{2N}$ and Bob's measurements are in the same basis but for $\theta=\frac{\pi V}{2N}$ then the correlations they get satisfy 
\begin{equation}
	I_{N}^{\star}=2N\mathrm{sin}^{2}\frac{\pi}{4N}<\frac{\pi^{2}}{8N}.\label{eq:violation}
\end{equation}
This implies that $I_{N}^{\star}$ can be made arbitrarily small for sufficiently large $N$.%
\footnote{However, as $N$ gets larger it becomes difficult to close the detection loophole \cite{scarani2009black} in the performed experiments, which is essential for any protocol that is based on non-local correlations. %
}

In our proof we will assume that the systems violate the chained Bell inequality. This is of course not the only possible choice for QKD protocols, although it is the most common one. Moreover, note that since for these type of systems PA is impossible, we cannot treat in general any system which produces some secrecy as a black box and therefore PA in general is impossible. 

\subsubsection*{Non-signalling systems}

Denote Alice's and Bob's system by $P_{XY|UV}$. A minimal requirement needed for any useful system is that Alice cannot signal to Bob using the system and vice versa, otherwise, the measured Bell violation will have no meaning. This can be ensured by placing Alice and Bob in space-like separated regions or by shielding their systems.
\begin{defn}
\label{Alice-&-Bob-n.s.}(Non-signalling between Alice and Bob). A
	$2n$-party conditional probability distribution $P_{XY|UV}$ over $X,Y,U,V\in\{0,1\}^{n}$ does not allow for signalling from Alice to Bob if 
	\begin{eqnarray*}
		\forall y,u,u',v\qquad\qquad\\
		\underset{x}{\sum}P_{XY|UV}(x,y|u,v) & = & \underset{x}{\sum}P_{XY|UV}(x,y|u',v)
	\end{eqnarray*}
	and does not allow for signalling from Bob to Alice if 
	\begin{eqnarray*}
		\forall x,v,v',u\qquad\qquad\\
		\underset{y}{\sum}P_{XY|UV}(x,y|u,v) & = & \underset{y}{\sum}P_{XY|UV}(x,y|u,v')\,.
	\end{eqnarray*}
\end{defn}

This definition implies that Bob (Alice) cannot infer from his (her) part of the system which input was given by Alice (Bob). The marginal system each of them sees is the same for all inputs of the other party and therefore the system $P_{XY|UV}$ cannot be used for signalling. 

In this paper we consider the conditions that we call time-ordered non-signalling conditions.

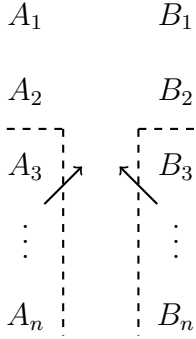
\begin{figure}[H]
\begin{centering}

\begin{tikzpicture} [font=\large]
		\draw (0,0) node {$A_{n}$};
		\draw (0,0.8) node {.};
		\draw (0,1) node {.};
		\draw (0,1.2) node {.};
		\draw (0,2) node {$A_{3}$};
		\draw (0,3) node {$A_{2}$};
		\draw (0,4) node {$A_{1}$};

		\draw (2,0) node {$B_{n}$};
		\draw (2,0.8) node {.};
		\draw (2,1) node {.};
		\draw (2,1.2) node {.};
		\draw (2,2) node {$B_{3}$};
		\draw (2,3) node {$B_{2}$};
		\draw (2,4) node {$B_{1}$};

		\draw[dashed, thick] (-0.25,2.5) -- (0.5,2.5);
		\draw[dashed, thick] (0.5,2.5) -- (0.5,-0.25);
		\draw[->, thick] (0.25,1.5) -- (0.75,2);

		\draw[dashed, thick] (1.5,2.5) -- (2.25,2.5);
		\draw[dashed, thick] (1.5,2.5) -- (1.5,-0.25);
		\draw[->, thick] (1.75,1.5) -- (1.25,2);
\end{tikzpicture}
\par\end{centering}

\caption{\label{fig:Time-orderd-non-signalling-condi}Time-ordered non-signalling condition for $i=3$. Signalling is impossible in the direction of the straight arrow.}
\end{figure}

\begin{defn}
\label{time-orderd-non-signalling}(Time-ordered non-signalling system).
	For any $i\in\{2,...,n\}$ denote the set $\{1,...,i-1\}$ by $I_{1}$ and the set $\{i,...,n\}$ by $I_{2}$ and for $i=1$ $I_{1}=\phi$ and $I_{2}=[n]$. A $2n$-party conditional probability distribution $P_{XY|UV}$ over $X,Y,U,V\in\{0,1\}^{n}$ is a time-ordered non-signalling system (does not allow for signalling from the future to the past) if for any $i\in[n]$, 
	\begin{multline*}
		\forall x_{I_{1}},y,u_{I_{1}},u_{I_{2}},u'_{I_{2}},v\qquad\qquad\\
		\shoveleft{\underset{x_{I_{2}}}{\sum}P_{XY|UV}(x_{I_{1}},x_{I_{2}},y|u_{I_{1}},u_{I_{2}},v)=}\\
		\shoveright{\underset{x_{I_{2}}}{\sum}P_{XY|UV}(x_{I_{1}},x_{I_{2}},y|u_{I_{1}},u'_{I_{2}},v)}\\
		\shoveleft{\forall x,y_{I_{1}},u,v_{I_{1}},v_{I_{2}},v'_{I_{2}}}\\
		\shoveleft{\underset{y_{I_{2}}}{\sum}P_{XY|UV}(x,y_{I_{1}},y_{I_{2}}|u,v_{I_{1}},v_{I_{2}})=}\\
		\underset{y_{I_{2}}}{\sum}P_{XY|UV}(x,y_{I_{1}},y_{I_{2}}|u,v_{I_{1}},v'_{I_{2}})\,.
	\end{multline*}
\end{defn}

Figure \ref{fig:Time-orderd-non-signalling-condi} illustrates these conditions. Note that the conditions of Definition \ref{Alice-&-Bob-n.s.} follow from these conditions.

\subsubsection*{Non-signalling adversaries}

When modelling a non-signalling adversary, the question in mind is as follows: given a system $P_{XY|UV}$ shared by Alice and Bob, for which some arbitrary non-signalling conditions hold, which extensions to a system $P_{XYZ|UVW}$, including the adversary Eve, are possible? The only principle which limits Eve is the non-signalling principle, which means that for any of her measurements $w$, the conditional system $P_{XY|UV}^{Z(w)=z}$ , for any $z\in Z$, must fulfil all of the non-signalling conditions that $P_{XY|UV}$ fulfils, and in addition $P_{XYZ|UVW}$ cannot allow signalling between Alice and Bob together and Eve.

We adopt here the model given in \cite{hanggi2009quantum,hanggi2010impossibility,hanggi2010device} of non-signalling adversaries. Because Eve cannot signal to Alice and Bob (even together) by her choice of input, we must have, for all $x,y,u,v,w,w'$, 
\begin{multline*}
	\sum_{z}P_{XYZ|UVW}(x,y,z|u,v,w)=\\
	\sum_{z}P_{XYZ|UVW}(x,y,z|u,v,w')=\\
	P_{XY|UV}(x,y|u,v).
\end{multline*}
We can therefore see Eve’s input as a choice of a convex decomposition of Alice’s and Bob’s system and her output as indicating one part of this decomposition. Formally, we can define every strategy of Eve
as a partition of Alice's and Bob's system in the following way. 

\begin{defn}
\label{Partition-system}
(Partition of the system). A partition of a given multipartite system $P_{XY|UV}$, which fulfils a certain set of non-signalling conditions $\mathcal{C}$, is a family of pairs $(p^{z},P_{XY|UV}^{z})$, where:
\begin{enumerate}
	\item $p^{z}$ is a classical distribution (i.e. for all $z$ $p^{z}\geq0$ and $\underset{z}{\sum}\, p^{z}=1$).
	\item For all $z$, $P_{XY|UV}^{z}$ is a system that fulfils $\mathcal{C}$.
	\item $P_{XY|UV}=\underset{z}{\sum}\, p^{z}P_{XY|UV\,.}^{z}$
\end{enumerate}
\end{defn}

In our scenario the goal of the adversary is to gain information about $f(X)$, for some function%
\footnote{It is enough to consider the case where Alice wants to create just one secret bit. An impossibility result for one bit implies an impossibility result for several bits.%
} $f:\{0,1\}^{n}\rightarrow\{0,1\}$. Note that since the adversarial strategy can be chosen after all public communication between Alice and Bob is done any additional random seed cannot help Alice and Bob. Therefore it is enough to consider deterministic functions in this case. 

In order to quantify how good a strategy $w$ is, i.e., how much information Eve gains about $f(X)$ by using $w$, we use the variational distance between the real system and the ideal system, in which $f(X)$ is uniformly distributed and independent of the adversary's system. 
\begin{lem}
	(Lemma 3.7 in \cite{hanggi2010device}). For the case $K=f(X)$, where $f:\{0,1\}^{n}\rightarrow\{0,1\}$, $U=u$, $V=v$, and where the strategy $w$ is defined by the partition $\left\{ (p^{z},P_{XY|UV}^{z})\right\} _{z\in\{0,1\}}$ such that $\Pr[K=0|Z=0]\geq\frac{1}{2}$, the distance from uniform of $f(X)$ given the strategy $w$ is 
	\begin{multline*}
		d\left(K|Z(w)\right)=\\
		p^{z=0}\cdot\left(\Pr[K=0|Z=0]-\Pr[K=1|Z=0]\right)\\
		-\frac{1}{2}\left(\Pr[K=0]-\Pr[K=1]\right).
	\end{multline*}
\end{lem}

\section{Main result}

In order to show an impossibility result we give a concrete adversarial strategy against any almost balanced hash functions. Eve will create a time-ordered non-signalling system between Alice, Bob and herself, such that when she inputs the hash function $f$ which was chosen by Alice on her side of the system, the output will be a guess at $f(x)$. We prove that this guess is correct with probability of at least $\frac{1}{2}+\frac{c}{n}$, where $c$ is some constant and $n$ is the number of systems shared by Alice and Bob. Against functions which are not almost balanced Eve can just use a trivial strategy and guess the value of the function without using her part of the system at all. 

As noted before, in order to prove an impossibility result it is enough to prove it for a specific system. We assume that when the adversary is not present, Alice and Bob share $n$ independent maximally entangled states and perform the measurements which achieve the violation of Equation (\ref{eq:violation}). We denote the system of each entangled pair by $P_{X_{i}Y_{i}|U_{i}V_{i}}$ for $i\in[n]$ and the whole system by $P_{XY|UV}=\underset{i\in[n]}{\prod}P_{X_{i}Y_{i}|U_{i}V_{i}}$. 

Let $f:\{0,1\}^{n}\rightarrow\{0,1\}$ be an almost balanced function. Showing a strategy is giving a partition of Alice's and Bob's system, as in Definition \ref{Partition-system}. Our partition will have 2 parts, $P_{XY|UV}^{0}$ and $P_{XY|UV}^{1}$, each occurring with probability $\frac{1}{2}$ and $P_{XY|UV}=\frac{1}{2}P_{XY|UV}^{0}+\frac{1}{2}P_{XY|UV}^{1}$. In our partition $P_{XY|UV}^{0}$ is biased towards $f(x)=0$ and $P_{XY|UV}^{1}$ towards $f(x)=1$. That is., if Eve gets an outcome of $z=0$ ($1$) when measuring her part of the system she knows that Alice's output $x$ is more likely to be a preimage of $0$ ($1$) according to $f$.

In this section we explain the idea and the intuition behind the adversarial strategy and the main principles of the proof. For the formal proof and technical details please see Appendix \ref{sec:Formal-strategy}. We start by describing how Eve can bias the system towards $f(x)=0$, i.e., what is $P_{XY|UV}^{0}$. 

Assume for the moment that for some given prefix of $x$, $x_{1...i-1}$, and function $f$ we have 
\begin{multline*}
	\underset{x_{i+1...n}}{\Pr}\left[f\left(x_{1...i-1}0x_{i+1...n}\right)=0\right] > \\
	\underset{x_{i+1...n}}{\Pr}\left[f\left(x_{1...i-1}1x_{i+1...n}\right)=0\right].
\end{multline*}
This implies that, for this specific prefix $x_{1...i-1}$, if Eve can guess the i'th bit $x_{i}$ then she can also guess the output bit of $f$. Therefor Eve can definitely benefit from biasing the  i'th bit towards 0.

Can the i'th subsystem be biased without changing the correlations Alice and Bob observe? The following lemma answers this question. 
\begin{lem}
\label{lem:individual-bias}
	For any $i\in[n]$, the system $P_{X_{i}Y_{i}|U_{i}V_{i}}$, for which $I_{N}\left(P_{X_{i}Y_{i}|U_{i}V_{i}}\right)=I_{N}^{\star}$, can be biased towards 0 (or 1) by $c(I_{N}^{\star})=\frac{I_{N}^{\star}}{2N}$.
\end{lem}
We denote the biased system by $P_{X_{i}Y_{i}|U_{i}V_{i}}^{z_{i}=\sigma}$ for $\sigma\in\{0,1\}$. The biased system is given in Appendix \ref{sec:Proof-bias-of-single-box}. 

Therefore, in our adversarial strategy, if the value of the i'th bit $x_{i}$, given the prefix $x_{1...i-1}$, has enough influence over the outcome of $f$ (we will soon define how much is enough), although the suffix is unknown, then the i'th system is being biased by $c(I_{N}^{\star})$. Note that for any prefix $x_{1...i-1}$ a different system $P_{X_{i}Y_{i}|U_{i}V_{i}}$ should be biased.

Next we say how Eve determines which subsystem $P_{X_{i}Y_{i}|U_{i}V_{i}}$ to bias for every $x$. For every function $f$, index $i\in[n]$ and prefix $x_{1...i-1}$ define 
\begin{multline*}
	\varDelta_{i}(x_{1...i-1})\equiv\biggm|\underset{x_{i+1...n}}{\Pr}\left[f(x_{1...i-1}0x_{i+1...n})=0\right]-\\
	\underset{x_{i+1...n}}{\Pr}\left[f(x_{1...i-1}1x_{i+1...n})=0\right]\biggm|.
\end{multline*}
$\varDelta_{i}(x_{1...i-1})$ quantifies how much influence the i'th bit has over $f$ given the prefix $x_{1...i-1}$%
\footnote{Note that the influence towards $f(x)=0$ and $f(x)=1$ is the same.%
}. For every $x$, Eve will bias the subsystem with the pivotal index, as we now define.

\begin{defn}
\label{Pivotal-index}(Pivotal index%
	\footnote{The terms `pivotal' and `influence' are taken from the field of Boolean function analysis.%
	}). Given $f:\{0,1\}^{n}\rightarrow\{0,1\}$, for any $x$, the pivotal index $i(x)\in[n]$ is the smallest index such that $\varDelta_{i(x)}(x_{1...i-1})\geq\frac{2}{3n}$. 
\end{defn}

Consider for example the function presented in Figure \ref{fig:Binary-tree}. The pivotal indices are marked in the binary tree of the function by a circle. For strings $x$ with prefix $x_{1}=0$ the pivotal index
is $i(x)=2$, while for strings with prefixes $x_{1}x_{2}=10$ and $x_{1}x_{2}=11$ the pivotal index is $i(x)=3$. 

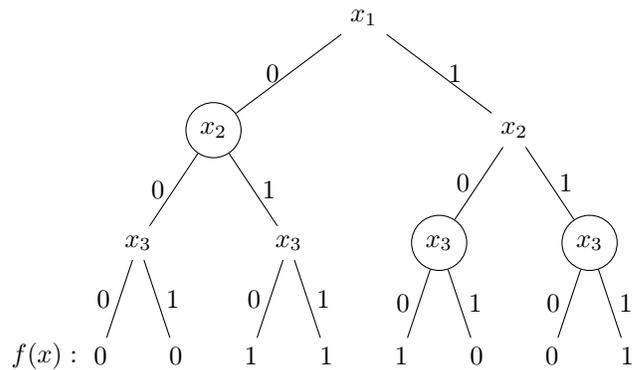
\begin{figure}[H]
\begin{centering}
\begin{tikzpicture}[level distance=1.5cm,
					level 1/.style={sibling distance=4cm},
					level 2/.style={sibling distance=2cm},
					level 3/.style={sibling distance=1cm},
					level 4/.style={sibling distance=0.5cm, level distance=0.75cm}]
	 \node {$x_1$}     
		child {node[circle,draw] {$x_2$} 
			child {node {$x_3$}
				child {node {0} 
					child [grow=left] {node {$f(x):$} edge from parent[draw=none]}
				edge from parent node[left] {0}}
				child {node {0} edge from parent node[right] {1}}
			edge from parent node[left] {0}}
			child {node {$x_3$}
				child {node {1} edge from parent node[left] {0}}
				child {node {1} edge from parent node[right] {1}}
			edge from parent node[right] {1}}
		edge from parent node[left] {0}}
		child {node {$x_2$}
			child {node[circle,draw] {$x_3$}
				child {node {1} edge from parent node[left] {0}}
				child {node {0} edge from parent node[right] {1}}
			edge from parent node[left] {0}}
			child {node[circle,draw] {$x_3$}
				child {node {0} edge from parent node[left] {0}}
				child {node {1} edge from parent node[right] {1}}
			edge from parent node[right] {1}}
		edge from parent node[right] {1}};
\end{tikzpicture}
\par\end{centering}

\caption{Binary tree with pivotal nodes. The pivotal nodes are marked with circles.\label{fig:Binary-tree}}
\end{figure}

Luckily, for every function $f:\{0,1\}^{n}\rightarrow\{0,1\}$ for which $|\underset{x}{\Pr}\left[f(x)=0\right]-\underset{x}{\Pr}\left[f(x)=1\right]|\leq\frac{1}{3}$ and every $x\in\{0,1\}^{n}$ there exists such a pivotal index $i(x)$ for which $\varDelta_{i(x)}(x_{1...i-1})\geq\frac{2}{3n}$ and therefore for every $x$ there exists some bit, $x_{i(x)}$, which can give non-negligible  information to Eve about the final output. 

\begin{lem}
\label{lem:pivotal-index}
	Let $f:\{0,1\}^{n}\rightarrow\{0,1\}$ be an almost balanced function, i.e. $|\underset{x}{\Pr}\left[f(x)=0\right]-\underset{x}{\Pr}\left[f(x)=1\right]|\leq\frac{1}{3}$. Then for any $x$ there exists a pivotal index $i(x)$ such that $\varDelta_{i(x)}(x_{1...i-1})\geq\frac{2}{3n}$.
\end{lem}

Lemma \ref{lem:pivotal-index} is proven in Appendix \ref{sec:Proof-pivotal-index}. 
Putting everything together, Eve's strategy is as follows. For every $x$ the $i(x)$'th subsystem, where $i(x)$ is the pivotal index of $x$, is biased. It is biased by $c(I_{N}^{\star})$ towards 0 if $\underset{x_{i+1...n}}{\Pr}\left[f(x_{1...i-1}0x_{i+1...n})=0\right]>\underset{x_{i+1...n}}{\Pr}\left[f(x_{1...i-1}1x_{i+1...n})=0\right]$ and towards 1 otherwise. The system $P_{XY|UV}^{0}$ which results from such a strategy is given in Equation (\ref{eq:formal-definition}) in Appendix \ref{sec:Formal-strategy}. 

The strategy for biasing the system towards $f(x)=1$ is symmetric to the strategy we described for $f(x)=0$. The only difference is that Eve will bias the i'th system by $c(I_{N}^{\star})$ towards 0 if $\underset{x_{i+1...n}}{\Pr}\left[f(x_{1...i-1}0x_{i+1...n})=0\right]<\underset{x_{i+1...n}}{\Pr}\left[f(x_{1...i-1}1x_{i+1...n})=0\right]$ and towards 1 otherwise, and not the other way around. The fact that these two symmetric systems put together a legal partition, as in Definition \ref{Partition-system}, is proven in Appendix \ref{sec:Formal-strategy}.

Since Eve biases a different subsystem $P_{X_{i(x)}Y_{i(x)}|U_{i(x)}V_{i(x)}}$ for every $x$, it is not clear that the system $P_{XY|UV}^{0}$ is indeed time-ordered non-signalling. The key idea for proving such a thing is that for every $x$, the location of the pivotal index $i(x)$ depends only on the prefix of $x$ until this index exactly. Intuitively,  in our case this corresponds to the fact that signalling is possible from past measurements to future measurements, or in other words, the fact that in any given time the prefix of $x$ can be saved in Alice's device. This is proven formally in Appendix \ref{sec:Formal-strategy}. 

How much information does this strategy give Eve? For every $x$ the $i(x)$'th subsystem is biased by $c(I_{N}^{\star})$. However, the advantage Eve gets from this shift in the probabilities is only $c(I_{N}^{\star})\cdot\varDelta_{i(x)}(x_{1...i-1})$ since the pivotal bit does not determine $f(x)$ exactly%
\footnote{When we shift some probability $\pi$ around from a cell which has probability $p_{1}$ to result in $f(x)=0$ (over the suffix) to a cell which has probability $p_{2}$ to result in $f(x)=0$ the advantage we get from shifting $\pi$ is $\pi\cdot(p_{2}-p_{1})$. In our case, $p_{2}-p_{1}$ is exactly $\varDelta_{i(x)}(x_{1...i-1})$ in our case. %
}.
Moreover, since $P_{XY|UV}^{0}$ and $P_{XY|UV}^{1}$ are symmetric and both occur with the same probability $\frac{1}{2}$ they both contribute the same amount of knowledge to Eve. 

As mentioned before, for any function for which $|\underset{x}{\Pr}\left[f(x)=0\right]-\underset{x}{\Pr}\left[f(x)=1\right]|>\frac{1}{3}$ Eve can just guess the value of the function with a constant success probability of at least $\frac{2}{3}$. Therefore these kind of functions do not bother us. Altogether we get the following theorem. 

\begin{thm}
	There exists a time-ordered non-signalling system $P_{XY|UV}$ as in Definition \ref{time-orderd-non-signalling} such that for any hash function $f:\{0,1\}^{n}\rightarrow\{0,1\}$ there exists a strategy $w$, for which the distance from uniform of $f(x)$ given $w$ is at least $c(I_{N}^{\star})\cdot\frac{2}{3n}$, i.e., $d\left(f(x)|Z(w)\right)\geq c(I_{N}^{\star})\cdot\frac{2}{3n}\in\varOmega(\frac{1}{n})$ where $I_{N}(P_{XY|UV})=I_{N}^{\star}$ and $c(I_{N}^{\star})=\frac{I_{N}^{\star}}{2N}$.%
		\footnote{Remember that $n$ is the number of systems while $N$ is the number of possible measurements for each system. For any given protocol $N$ is constant and therefore so also is $I_{N}^{\star}$.%
		}
\end{thm}

\begin{proof}
	If $f:\{0,1\}^{n}\rightarrow\{0,1\}$ is an almost balanced function as in Lemma \ref{lem:pivotal-index} then $w$ is the strategy described above, for which $d\left(f(x)|Z(w)\right)\geq c(I_{N}^{\star})\cdot\frac{2}{3n}$. Otherwise, the strategy is to guess $f(x)$. For this trivial strategy we have $d\left(f(x)|Z(w)\right)\geq\frac{2}{3}-\frac{1}{2}\geq c(I_{N}^{\star})\cdot\frac{2}{3n}$.
\end{proof}

\subsubsection*{Concluding remarks and open questions}

In this paper we showed that when considering systems which can signal only forward in time and non-signalling adversaries, then super-polynomial privacy amplification by any hash function is impossible. For protocols which are based on the violation of the chained Bell inequalities, we presented a specific adversarial strategy which uses the memory of the device in order to gain information about the value of the function. 

It is not yet clear whether our result is tight. We showed that, independently of which hash function Alice is using, Eve can bias the key by at least $\Omega\left(\frac{1}{n}\right)$. For some bad choices of hash functions Eve can get even more information than $\Omega\left(\frac{1}{n}\right)$ by using the same strategy. For example, if the chosen hash function is the XOR, then by using the exact same strategy, but with a different analysis, Eve can bias the final key bit by a constant. When using the Majority function this strategy can only give her $\Omega\left(\frac{1}{\sqrt{n}}\right)$ bias. Is this the best Eve can do? Can we find a specific hash function for which she cannot do any better than this? The question whether linear privacy amplification is possible or not therefore remains open.

\paragraph*{Acknowledgments:}

Rotem Arnon-Friedman thanks Roger Colbeck for helpful comments. Both authors acknowledge support from the FP7 FET-Open Project QCS. 

\end{multicols}

\begin{center}
\line(1,0){250}
\end{center}

\begin{multicols}{2}
\bibliographystyle{unsrt}
\bibliography{/Users/rotem/Documents/Research/biblopropos.bib}
\end{multicols}

\begin{center}
\line(1,0){250}
\end{center}

\appendixpage
\begin{appendices}

\section{\label{sec:Proof-bias-of-single-box}Proof of Lemma \ref{lem:individual-bias}}

We now prove the following lemma:
\begin{lem*}
	For any $i\in[n]$, the system $P_{X_{i}Y_{i}|U_{i}V_{i}}$, for which $I_{N}\left(P_{X_{i}Y_{i}|U_{i}V_{i}}\right)=I_{N}^{\star}$, can be biased towards 0 (or 1) by $c(I_{N}^{\star})=\frac{I_{N}^{\star}}{2N}$.
\end{lem*}

\begin{proof}

	In order to prove this we define the system $P_{X_{i}Y_{i}|U_{i}V_{i}}^{z_{i}=0}$ which is biased towards 0 by $c(I_{N}^{\star})$. We do so by shifting probabilities around in the original unbiased system $P_{X_{i}Y_{i}|U_{i}V_{i}}$. The original system $P_{X_{i}Y_{i}|U_{i}V_{i}}$, as in Figure \ref{fig:unbiased-box}, describes the measurements statistics of the maximally entangled state $|\Phi+\rangle=\frac{1}{\sqrt{2}}\left(|00\rangle+|11\rangle\right)$ in the basis $\left\{ \mathrm{cos}\frac{\theta}{2}|0\rangle+\mathrm{sin}\frac{\theta}{2}|1\rangle,\mathrm{sin}\frac{\theta}{2}|0\rangle-\mathrm{cos}\frac{\theta}{2}|1\rangle\right\} $, where for Alice $\theta=\frac{\pi U}{2N}$ , $U\in\{0,2,...2N-2\}$ and for Bob $\theta=\frac{\pi V}{2N}$, $V\in\{1,3,...2N-1\}$ . 

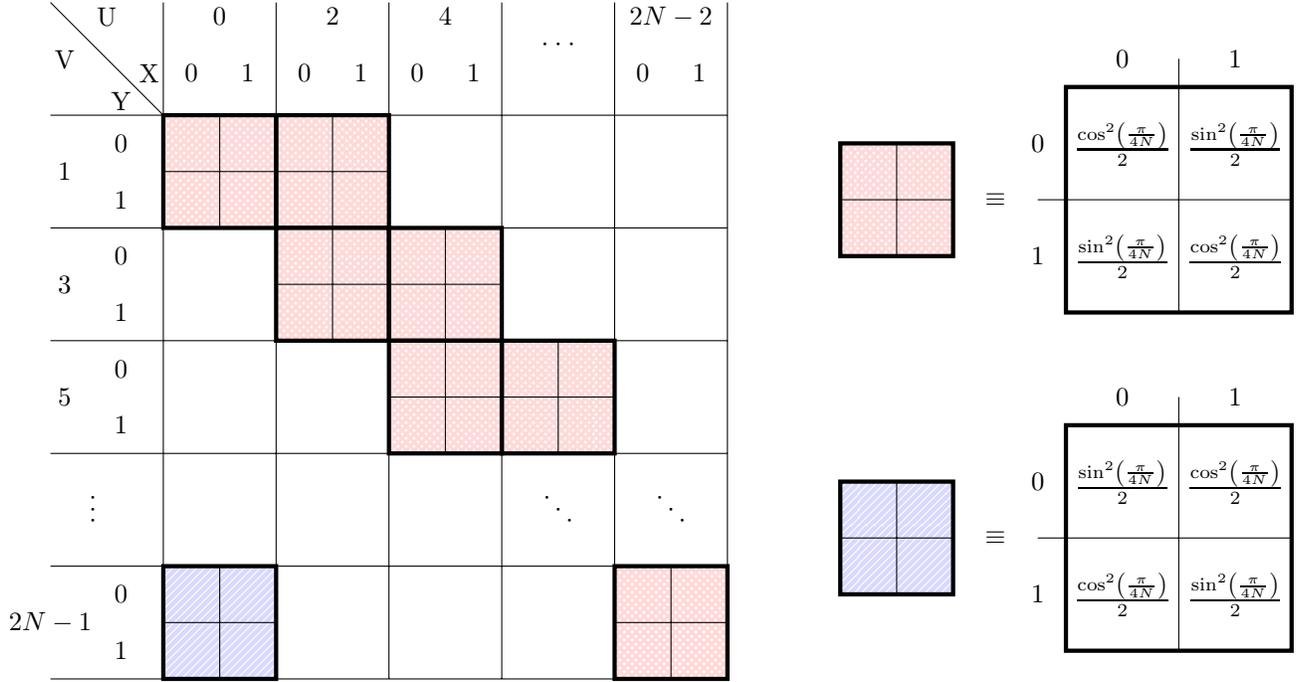
\begin{figure}[p]
\begin{centering}
\begin{tikzpicture}[scale=0.75]
	\draw[step=2] (0,0) grid (10,10);

	\draw (-2,12)--(0,10);
	\draw (-0.75,10.25) node {Y};
	\draw (-0.25,10.75) node {X};
	\draw (-1,11.75) node {U};
	\draw (-1.75,11) node {V};

	\draw (0,10)--(0,12);
	\draw (1,11.75) node {0};
		\draw (0.5,10.75) node {0};
		\draw (1.5,10.75) node {1};
	\draw (2,10)--(2,12);
	\draw (3,11.75) node {2};
		\draw (2.5,10.75) node {0};
		\draw (3.5,10.75) node {1};
	\draw (4,10)--(4,12);
	\draw (5,11.75) node {4};
		\draw (4.5,10.75) node {0};
		\draw (5.5,10.75) node {1};
	\draw (6,10)--(6,12);
	\draw (7,11.25) node {$\cdot\cdot\cdot$};
	\draw (8,10)--(8,12);
	\draw (9,11.75) node {$2N-2$};
		\draw (8.5,10.75) node {0};
		\draw (9.5,10.75) node {1};
	\draw (10,10)--(10,12);

	\draw (-2,0)--(0,0);
	\draw (-2,1) node {$2N-1$};
		\draw (-0.75,0.5) node {1};
		\draw (-0.75,1.5) node {0};
	\draw (-2,2)--(0,2);
	\draw (-1.25,2.8) node {$\cdot$};
	\draw (-1.25,3) node {$\cdot$};
	\draw (-1.25,3.2) node {$\cdot$};
	\draw (-2,4)--(0,4);
	\draw (-1.75,5) node {5};
		\draw (-0.75,4.5) node {1};
		\draw (-0.75,5.5) node {0};
	\draw (-2,6)--(0,6);
	\draw (-1.75,7) node {3};
		\draw (-0.75,6.5) node {1};
		\draw (-0.75,7.5) node {0};
	\draw (-2,8)--(0,8);
	\draw (-1.75,9) node {1};
		\draw (-0.75,8.5) node {1};
		\draw (-0.75,9.5) node {0};
	\draw (-2,10)--(0,10);

	\fill[red!15, postaction={pattern=crosshatch dots, pattern color=white}](0,8) rectangle (4,10);
	\draw[step=1] (0,8) grid (4,10);
	\draw[ultra thick] (0,8) rectangle (2,10);
	\draw[ultra thick] (2,8) rectangle (4,10);

	\fill[red!15, postaction={pattern=crosshatch dots, pattern color=white}](2,6) rectangle (6,8);
	\draw[step=1] (2,6) grid (6,8);
	\draw[ultra thick] (2,6) rectangle (4,8);
	\draw[ultra thick] (4,6) rectangle (6,8);

	\fill[red!15, postaction={pattern=crosshatch dots, pattern color=white}](4,4) rectangle (8,6);
	\draw[step=1] (4,4) grid (8,6);
	\draw[ultra thick] (4,4) rectangle (6,6);
	\draw[ultra thick] (6,4) rectangle (8,6);

	\draw (8.8,3.2) node {$\cdot$};
	\draw (9,3) node {$\cdot$};
	\draw (9.2,2.8) node {$\cdot$};

	\draw (6.8,3.2) node {$\cdot$};
	\draw (7,3) node {$\cdot$};
	\draw (7.2,2.8) node {$\cdot$};

	\fill[red!15, postaction={pattern=crosshatch dots, pattern color=white}](8,0) rectangle (10,2);
	\draw[step=1] (8,0) grid (10,2);
	\draw[ultra thick] (8,0) rectangle (10,2);

	\fill[blue!15, postaction={pattern=north east lines, pattern color=white}](0,0) rectangle (2,2);
	\draw[step=1] (0,0) grid (2,2);
	\draw[ultra thick] (0,0) rectangle (2,2);

	\begin{scope}[shift={(12,6.5)}]
		\fill[red!15, postaction={pattern=crosshatch dots, pattern color=white}](0,1) rectangle (2,3);
		\draw[step=1] (0,1) grid (2,3);
		\draw[ultra thick] (0,1) rectangle (2,3);
		
		\draw (2.75,2) node {$\equiv$};

		\draw[step=2] (4,0) grid (8,4);
		\draw[ultra thick] (4,0) rectangle (8,4);
		
		\draw (5,4.5) node {0};
		\draw (6,4)--(6,4.5);
		\draw (7,4.5) node {1};

		\draw (3.5,3) node {0};
		\draw (4,2)--(3.5,2);
		\draw (3.5,1) node {1};

		\draw (5,3) node {$\frac{\mathrm{cos}^{2}\left(\frac{\pi}{4N}\right)}{2}$};
		\draw (5,1) node {$\frac{\mathrm{sin}^{2}\left(\frac{\pi}{4N}\right)}{2}$};
		\draw (7,3) node {$\frac{\mathrm{sin}^{2}\left(\frac{\pi}{4N}\right)}{2}$};
		\draw (7,1) node {$\frac{\mathrm{cos}^{2}\left(\frac{\pi}{4N}\right)}{2}$};
	\end{scope}

	\begin{scope}[shift={(12,0.5)}]
		\fill[blue!15, postaction={pattern=north east lines, pattern color=white}](0,1) rectangle (2,3);
		\draw[step=1] (0,1) grid (2,3);
		\draw[ultra thick] (0,1) rectangle (2,3);
		
		\draw (2.75,2) node {$\equiv$};

		\draw[step=2] (4,0) grid (8,4);
		\draw[ultra thick] (4,0) rectangle (8,4);
		
		\draw (5,4.5) node {0};
		\draw (6,4)--(6,4.5);
		\draw (7,4.5) node {1};

		\draw (3.5,3) node {0};
		\draw (4,2)--(3.5,2);
		\draw (3.5,1) node {1};

		\draw (5,3) node {$\frac{\mathrm{sin}^{2}\left(\frac{\pi}{4N}\right)}{2}$};
		\draw (5,1) node {$\frac{\mathrm{cos}^{2}\left(\frac{\pi}{4N}\right)}{2}$};
		\draw (7,3) node {$\frac{\mathrm{cos}^{2}\left(\frac{\pi}{4N}\right)}{2}$};
		\draw (7,1) node {$\frac{\mathrm{sin}^{2}\left(\frac{\pi}{4N}\right)}{2}$};
	\end{scope}
	  
\end{tikzpicture}
\par\end{centering}

\caption{\label{fig:unbiased-box}The unbiased system $P_{X_{i}Y_{i}|U_{i}V_{i}}$ for which $I_{N}^{\star}=2N\mathrm{sin}^{2}\frac{\pi}{4N}$. The empty squares in the figure are not relevant for the correlations and therefore are not considered in cryptographic protocols. }
\end{figure}

	In order to bias this system towards 0 we shift probabilities within each individual square in the figure, such that each square will be biased toward 0 by $\mathrm{sin}^{2}\left(\frac{\pi}{4N}\right)$. We do so by shifting in every row probability of $\frac{1}{2}\mathrm{sin}^{2}\left(\frac{\pi}{4N}\right)$ out from the cell with $x_{i}=1$ and into the cell with $x_{i}=0$, as indicated in Figure \ref{fig:biased-box}. Each square corresponds to a different measurement made by Alice and Bob, and therefore for every measurement the bias is the same and equivalent to $c(I_{N}^{\star})=\frac{1}{2N}\cdot I_{N}^{\star}$.

\begin{figure}[p]
\begin{centering}
\begin{tikzpicture}[scale=0.85]

	\begin{scope}
		\fill[red!15, postaction={pattern=crosshatch dots, pattern color=white}](0,0) rectangle (4,4);
		\draw[step=2] (0,0) grid (4,4);
		\draw[ultra thick] (0,0) rectangle (4,4);
		
		\draw (1,4.5) node {0};
		\draw (2,4)--(2,4.5);
		\draw (3,4.5) node {1};

		\draw (-0.5,3) node {0};
		\draw (0,2)--(-0.5,2);
		\draw (-0.5,1) node {1};

		\draw (1,3) node {$\frac{\mathrm{cos}^{2}\left(\frac{\pi}{4N}\right)}{2}$};
		\draw (1,1) node {$\frac{\mathrm{sin}^{2}\left(\frac{\pi}{4N}\right)}{2}$};
		\draw (3,3) node {$\frac{\mathrm{sin}^{2}\left(\frac{\pi}{4N}\right)}{2}$};
		\draw (3,1) node {$\frac{\mathrm{cos}^{2}\left(\frac{\pi}{4N}\right)}{2}$};

		\draw (5,2) node[font=\large] {$\Longrightarrow$};

		\fill[red!15, postaction={pattern=crosshatch dots, pattern color=white}](6,0) rectangle (10,4);
		\draw[step=2] (6,0) grid (10,4);
		\draw[ultra thick] (6,0) rectangle (10,4);
		
		\draw (7,4.5) node {0};
		\draw (8,4)--(8,4.5);
		\draw (9,4.5) node {1};

		\draw (5.5,3) node {0};
		\draw (6,2)--(5.5,2);
		\draw (5.5,1) node {1};

		\draw (7,3) node {$\frac{1}{2}$};
		\draw (7,1) node {$\mathrm{sin}^{2}\left(\frac{\pi}{4N}\right)$};
		\draw (9,3) node {$0$};
		\draw (9,1.5) node {$\frac{\mathrm{cos}^{2}\left(\frac{\pi}{4N}\right)}{2}-$};
		\draw (9,0.5) node {$\frac{\mathrm{sin}^{2}\left(\frac{\pi}{4N}\right)}{2}$};
	\end{scope}
	
	\begin{scope}[shift={(0,-6)}]
		\fill[blue!15, postaction={pattern=north east lines, pattern color=white}](0,0) rectangle (4,4);
		\draw[step=2] (0,0) grid (4,4);
		\draw[ultra thick] (0,0) rectangle (4,4);
		
		\draw (1,4.5) node {0};
		\draw (2,4)--(2,4.5);
		\draw (3,4.5) node {1};

		\draw (-0.5,3) node {0};
		\draw (0,2)--(-0.5,2);
		\draw (-0.5,1) node {1};

		\draw (1,3) node {$\frac{\mathrm{sin}^{2}\left(\frac{\pi}{4N}\right)}{2}$};
		\draw (1,1) node {$\frac{\mathrm{cos}^{2}\left(\frac{\pi}{4N}\right)}{2}$};
		\draw (3,3) node {$\frac{\mathrm{cos}^{2}\left(\frac{\pi}{4N}\right)}{2}$};
		\draw (3,1) node {$\frac{\mathrm{sin}^{2}\left(\frac{\pi}{4N}\right)}{2}$};

		\draw (5,2) node[font=\large] {$\Longrightarrow$};

		\fill[blue!15, postaction={pattern=north east lines, pattern color=white}](6,0) rectangle (10,4);
		\draw[step=2] (6,0) grid (10,4);
		\draw[ultra thick] (6,0) rectangle (10,4);
		
		\draw (7,4.5) node {0};
		\draw (8,4)--(8,4.5);
		\draw (9,4.5) node {1};

		\draw (5.5,3) node {0};
		\draw (6,2)--(5.5,2);
		\draw (5.5,1) node {1};

		\draw (7,3) node {$\mathrm{sin}^{2}\left(\frac{\pi}{4N}\right)$};
		\draw (7,1) node {$\frac{1}{2}$};
		\draw (9,3.5) node {$\frac{\mathrm{cos}^{2}\left(\frac{\pi}{4N}\right)}{2}-$};
		\draw (9,2.5) node {$\frac{\mathrm{sin}^{2}\left(\frac{\pi}{4N}\right)}{2}$};
		\draw (9,1) node {$0$};
	\end{scope}
	  	  
\end{tikzpicture}
\par\end{centering}

\caption{\label{fig:biased-box}The biased system $P_{X_{i}Y_{i}|U_{i}V_{i}}^{z_{i}=0}$. Here are the same squares of Figure \ref{fig:unbiased-box} after the probability shift.}
\end{figure}
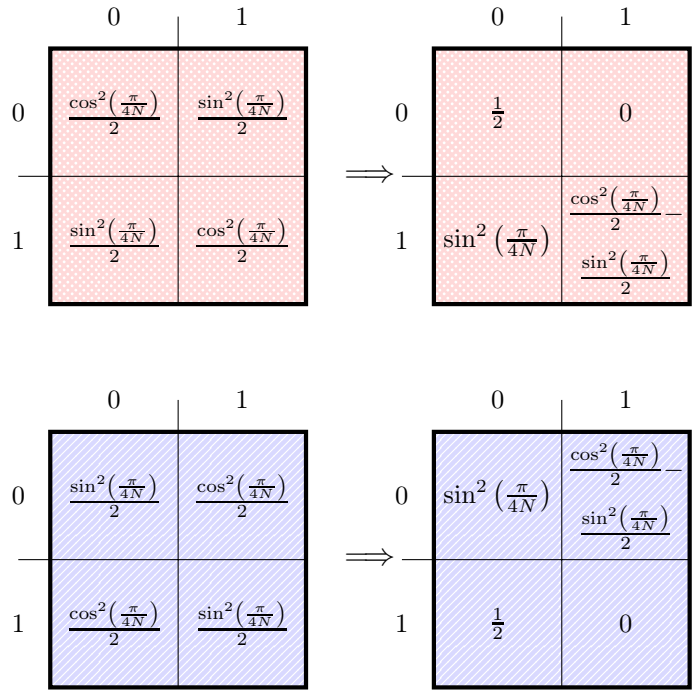

	Note that by shifting probabilities in this way we do not change the correlations of the system, i.e., $I_{N}\left(P_{X_{i}Y_{i}|U_{i}V_{i}}^{z_{i}=0}\right)=I_{N}^{\star}$. 

	The system $P_{X_{i}Y_{i}|U_{i}V_{i}}^{z_{i}=1}$, which is biased towards 1, is symmetric. That is, we shift the same amount of probability but in the opposite direction (from $x_{i}=0$ to $x_{i}=1$). This also implies that $\frac{1}{2}P_{X_{i}Y_{i}|U_{i}V_{i}}^{z_{i}=0}+\frac{1}{2}P_{X_{i}Y_{i}|U_{i}V_{i}}^{z_{i}=1}=P_{X_{i}Y_{i}|U_{i}V_{i}}$.

\end{proof}

\section{\label{sec:Proof-pivotal-index}Proof of Lemma \ref{lem:pivotal-index}}

For convenience we rewrite Lemma \ref{lem:pivotal-index} here again. 
\begin{lem*}
	Let $f:\{0,1\}^{n}\rightarrow\{0,1\}$ be an almost balanced function for which $|\underset{x}{\Pr}\left[f(x)=0\right]-\underset{x}{\Pr}\left[f(x)=1\right]|\leq\frac{1}{3}.$ Then for any $x$ there exists a pivotal index $i(x)$ such that $\varDelta_{i(x)}(x_{1...i-1})\geq\frac{2}{3n}$, where 
	\[
	\varDelta_{i}(x_{1...i-1})\equiv\biggm|\underset{x_{i+1...n}}{\Pr}\left[f(x_{1...i-1}0x_{i+1...n})=0\right]-\underset{x_{i+1...n}}{\Pr}\left[f(x_{1...i-1}1x_{i+1...n})=0\right]\biggm|.
	\]
\end{lem*}

\begin{proof}
	Let $\pi^{0}(x_{1...i-1})=\underset{x_{i...n}}{\Pr}\left[f(x)=0\right]$ where $x=x_{1...i-1}x_{i...n}$ and note the following properties:
	\begin{eqnarray*}
		\pi^{0}(x_{1...i-1}) & = & \frac{1}{2}\left[\pi^{0}(x_{1...i-1}0)+\pi^{0}(x_{1...i-1}1)\right]\\
		\pi^{0}(\phi) & \geq & \frac{1}{3}\\
		\pi^{0}(x_{1...n}) & \in & \left\{ 0,1\right\} .
	\end{eqnarray*}
	Assume w.l.o.g $\pi^{0}(x_{1...n})=0$ (the proof is symmetric for the case $\pi^{0}(x_{1...n})=1$). 
	
	Let $\underset{i\in[n]}{max}|\pi^{0}(x_{1...i})-\pi^{0}(x_{1...i-1})|\leq\zeta$. This implies the following: 
	\[
	\frac{1}{3}\leq\biggm|\pi^{0}(\phi)-\pi^{0}(x_{1...n})\biggm|\leq n\cdot\zeta
	\]
	and therefore $\zeta\geq\frac{1}{3n}$. I.e., there exists $j\in[n]$ such that $|\pi^{0}(x_{1...j})-\pi^{0}(x_{1...j-1})|\geq\frac{1}{3n}$ and since we assumed $\pi^{0}(x_{1...n})=0$ we can farther write $\pi^{0}(x_{1...j-1})\geq\pi^{0}(x_{1...j})+\frac{1}{3n}$. Moreover, since 
	\begin{eqnarray*}
		\pi^{0}(x_{1...j-1}) & = & \frac{1}{2}\left[\pi^{0}(x_{1...j-1}0)+\pi^{0}(x_{1...j-1}1)\right]\\
		 & = & \frac{1}{2}\left[\pi^{0}(x_{1...j-1}x_{j})+\pi^{0}(x_{1...j-1}\overline{x_{j}})\right]
	\end{eqnarray*}
	we get that $\pi^{0}(x_{1...j-1}\overline{x_{j}})\geq\pi^{0}(x_{1...j-1}x_{j})+\frac{2}{3n}$ and therefore for any $x$ there exists an index $i(x)=j$ for which $\varDelta_{i(x)}(x_{1...i-1})\geq\frac{2}{3n}$.
\end{proof}

\section{\label{sec:Formal-strategy}Formal definition of the strategy}

As explained in the main text, Eve's strategy is to use a partition $\left\{ \left(\frac{1}{2},P_{XY|UV}^{z}\right)\right\} _{z\in\{0,1\}}$ for which $P_{XY|UV}=\frac{1}{2}P_{XY|UV}^{0}+\frac{1}{2}P_{XY|UV}^{1}$. The systems $P_{XY|UV}^{0}$ and $P_{XY|UV}^{1}$ are obtained by biasing one individual subsystem $P_{X_{i(x)}Y_{i(x)}|U_{i(x)}V_{i(x)}}$ for each $x$. For any $i\in[n]$ let $P_{X_{i}Y_{i}|U_{i}V_{i}}^{z_{i}=0}$ and $P_{X_{i}Y_{i}|U_{i}V_{i}}^{z_{i}=1}$ be the biased systems as defined in Appendix \ref{sec:Proof-bias-of-single-box}. The system $P_{XY|UV}^{0}$ is then formally defined by 
\begin{multline}
	P_{XY|UV}^{0}(x,y|u,v)=\overset{i(x)-1}{\underset{j=1}{\prod}}P_{X_{j}Y_{j}|U_{j}V_{j}}(x_{j},y_{j}|u_{j},v_{j})\cdot P_{X_{i(x)}Y_{i(x)}|U_{i(x)}V_{i(x)}}^{z_{i}=\sigma}(x_{i(x)},y_{i(x)}|u_{i(x)},v_{i(x)})\cdot\\
	\overset{n}{\underset{j=i(x)+1}{\prod}}P_{X_{j}Y_{j}|U_{j}V_{j}}(x_{j},y_{j}|u_{j},v_{j})\label{eq:formal-definition}
\end{multline}
where $i(x)$ is the pivotal index of $x$ as in Definition \ref{Pivotal-index} and 
\[
\sigma=\begin{cases}
0 & \underset{x_{i+1...n}}{\Pr}\left[f(x_{1...i-1}0x_{i+1...n})=0\right]>\underset{x_{i+1...n}}{\Pr}\left[f(x_{1...i-1}1x_{i+1...n})=0\right]\\
1 & \mbox{otherwise}
\end{cases}.
\]
That is, if $f(x)$ is more likely to result in $f(x)=0$ if $x_{i(x)}=0$ then Eve biases the $i(x)$'th system towards $0$ and if not then towards 1. Note that since Eve manipulates the $i(x)$'th system only if $\varDelta_{i(x)}(x_{1...i-1})\geq\frac{2}{3n}$ $\underset{x_{i+1...n}}{\Pr}\left[f(x_{1...i-1}0x_{i+1...n})=0\right]$ never equals $\underset{x_{i+1...n}}{\Pr}\left[f(x_{1...i-1}1x_{i+1...n})=0\right]$. 

The complementary system $P_{XY|UV}^{1}$ is defined in the exact same way but with $\overline{\sigma}$ instead of $\sigma$.

In order to prove the legality of the strategy we first prove that $P_{XY|UV}^{0}$ is a probability distribution.
\begin{lem*}
	The system $P_{XY|UV}^{0}$ is a probability distribution. That is, 
	\begin{enumerate}
		\item For all $x,y,u,v$ $P_{XY|UV}^{0}(x,y|u,v)\geq0$
		\item The system is normalized. For all $u,v$, $\underset{x,y}{\sum}P_{XY|UV}^{0}(x,y|u,v)=1$
	\end{enumerate}
\end{lem*}

\begin{proof}
	Each of the multiplicands in Equation (\ref{eq:formal-definition}) is non-negative and therefore for all $x,y,u,v$ it also holds that $P_{XY|UV}^{0}(x,y|u,v)\geq0$. Farther more, since 
	\begin{multline}
	P_{X_{i(x)}Y_{i(x)}|U_{i(x)}V_{i(x)}}^{z_{i}=\sigma}(x_{i(x)},y_{i(x)}|u_{i(x)},v_{i(x)})+P_{X_{i(x)}Y_{i(x)}|U_{i(x)}V_{i(x)}}^{z_{i}=\sigma}(\overline{x_{i(x)}},y_{i(x)}|u_{i(x)},v_{i(x)})=\\
	P_{X_{i(x)}Y_{i(x)}|U_{i(x)}V_{i(x)}}(x_{i(x)},y_{i(x)}|u_{i(x)},v_{i(x)})+P_{X_{i(x)}Y_{i(x)}|U_{i(x)}V_{i(x)}}(\overline{x_{i(x)}},y_{i(x)}|u_{i(x)},v_{i(x)})\label{eq:property-1}
	\end{multline}
	(cf. Figure \ref{fig:biased-box}) we also have 
	\[
	P_{XY|UV}^{0}(x,y|u,v)+P_{XY|UV}^{0}(x^{i(x)},y|u,v)=P_{XY|UV}(x,y|u,v)+P_{XY|UV}(x^{i(x)},y|u,v)
	\]
	where $x^{i(x)}$ is the string $x$ with the $i(x)$'th bit flipped, i.e.,$x^{i(x)}=x_{1}...x_{i(x)-1}\overline{x}_{i(x)}x_{i(x)+1}...x_{n}$. This implies that 
	\[
	\underset{x,y}{\sum}P_{XY|UV}^{0}(x,y|u,v)=\underset{x,y}{\sum}P_{XY|UV}(x,y|u,v)=1\,.
	\]
\end{proof}
The same proof holds for $P_{XY|UV}^{1}$ as well. The fact that $P_{XY|UV}^{0}$ and $P_{XY|UV}^{1}$ are probability distributions is not enough. We also need to prove that they are complementary systems, i.e., $P_{XY|UV}=\frac{1}{2}P_{XY|UV}^{0}+\frac{1}{2}P_{XY|UV}^{1}$.

\begin{lem*}
	$P_{XY|UV}=\frac{1}{2}P_{XY|UV}^{0}+\frac{1}{2}P_{XY|UV}^{1}$.
\end{lem*}

\begin{proof}
	For simplicity we drop the subscript $XY|UV$ from all the systems. For example $P(x,y|u,v)$ should be understood as $P_{XY|UV}(x,y|u,v)$ while $P^{z_{i}=\sigma}(x_{i(x)},y_{i(x)}|u_{i(x)},v_{i(x)})$ should be understood as $P_{X_{i(x)}Y_{i(x)}|U_{i(x)}V_{i(x)}}^{z_{i}=\sigma}(x_{i(x)},y_{i(x)}|u_{i(x)},v_{i(x)})$.
	\begin{eqnarray*}
	2P(x,y|u,v)-P^{0}(x,y|u,v) & = & 2\overset{n}{\underset{j=1}{\prod}}P(x_{j},y_{j}|u_{j},v_{j})-P^{0}(x,y|u,v)\\
	 & = & \overset{n}{\underset{\begin{array}{c}
	j=1\\
	j\neq i(x)
	\end{array}}{\prod}}P(x_{j},y_{j}|u_{j},v_{j})\cdot\left[2P(x_{i(x)},y_{i(x)}|u_{i(x)},v_{i(x)})-P^{z_{i}=\sigma}(x_{i(x)},y_{i(x)}|u_{i(x)},v_{i(x)})\right]\\
	 & = & \overset{n}{\underset{\begin{array}{c}
	j=1\\
	j\neq i(x)
	\end{array}}{\prod}}P(x_{j},y_{j}|u_{j},v_{j})\cdot P^{z_{i}=\overline{\sigma}}(x_{i(x)},y_{i(x)}|u_{i(x)},v_{i(x)})\\
	 & = & P^{1}(x,y|u,v)\,.\qedhere
	\end{eqnarray*}
\end{proof}

We have only left to show that the system $P_{XY|UV}^{0}$ is a time-ordered non-signalling system.
\begin{lem*}
	The system \textup{$P_{XY|UV}^{0}$} is time-ordered non-signalling as in Definition \ref{time-orderd-non-signalling}.
\end{lem*}

\begin{proof}
	For the conditions on Bob's side of the system we first note the following. In the system $P_{X_{i(x)}Y_{i(x)}|U_{i(x)}V_{i(x)}}^{z_{i}=\sigma}$ we shift probabilities only within the same row. Moreover, we shift the probability in the exact same way on identical rows (cf. Figure \ref{fig:biased-box}: the first row in the upper boxes is identical to the second row in the lower boxes). It then follows from Lemmas 4.4, 4.5 and 4.6 in \cite{arnon2012towards} that full non-signalling conditions hold for Bob's side (i.e., every subset of his systems cannot signal any other subset of systems). In particular, the time-ordered non-signalling conditions hold. 

	For simplicity we drop the subscript $XY|UV$ from all the systems as in the previous proof. We now want to prove that the conditions on Alice's side hold, i.e., that for any sets $I_{1},I_{2}$ as in Definition \ref{time-orderd-non-signalling} 
	\begin{equation}
		\forall x_{I_{1}},y,u_{I_{1}},u_{I_{2}},u'_{I_{2}},v\quad\underset{x_{I_{2}}}{\sum}P^{0}(x_{I_{1}},x_{I_{2}},y|u_{I_{1}},u_{I_{2}},v)=\underset{x_{I_{2}}}{\sum}P^{0}(x_{I_{1}},x_{I_{2}},y|u_{I_{1}},u'_{I_{2}},v).\label{eq:Alice-side-eq}
	\end{equation}
	For any $x_{I_{1}}$ there are two possible cases, as indicated in Figure \ref{fig:cases}; the pivotal index $i(x)$ is either in $I_{1}$ or in $I_{2}$. We show that on both cases the time-ordered non-signalling conditions on Alice's side hold. 

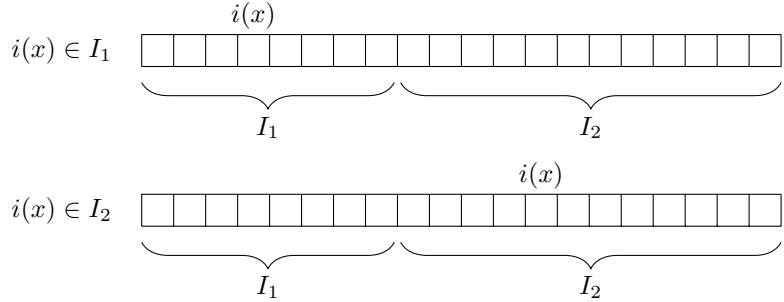
\begin{figure}[b]
\centering{}
\begin{tikzpicture}[scale=0.85]
	\begin{scope}
		\draw (-1.25,0.25) node {$i(x)\in I_{1}$};
		\draw[step=0.5] (0,0) grid (10,0.5);
		\draw[decorate,decoration={brace,mirror,amplitude=10pt}] (0,-0.25) -- (3.95,-0.25) node [midway,yshift=-0.6cm] {$I_{1}$};
		\draw[decorate,decoration={brace,mirror,amplitude=10pt}] (4.05,-0.25) -- (10,-0.25)  node [midway,yshift=-0.6cm] {$I_{2}$};

		\draw (1.75,0.8) node {$i(x)$};
	\end{scope}

	\begin{scope}[shift={(0,-2.5)}]
		\draw (-1.25,0.25) node {$i(x)\in I_{2}$};
		\draw[step=0.5] (0,0) grid (10,0.5);
		\draw[decorate,decoration={brace,mirror,amplitude=10pt}] (0,-0.25) -- (3.95,-0.25) node [midway,yshift=-0.6cm] {$I_{1}$};
		\draw[decorate,decoration={brace,mirror,amplitude=10pt}] (4.05,-0.25) -- (10,-0.25)  node [midway,yshift=-0.6cm] {$I_{2}$};

		\draw (6.25,0.8) node {$i(x)$};
	\end{scope}
\end{tikzpicture}\caption{\label{fig:cases}Two possible cases: $i(x)\in I_{1}$ or $i(x)\in I_{2}$}
\end{figure}

	First assume that for the pivotal index $i(x)\in I_{1}$. For any $u,u'$ and $v$, for any $x$ let 
	\[
	x_{j}'=\begin{cases}
	\overline{x_{j}} & u_{j}\neq u'_{j}\:\wedge v_{j}=2N-1\\
	x_{j} & \mbox{otherwise}
	\end{cases}
	\]
	and $x'=x'_{1}\cdot\cdot\cdot x'_{n}$. Furthermore, note that for the unbiased system $P_{XY|UV}$ we have $P_{XY|UV}(x,y|u',v)=P_{XY|UV}(x',y|u,v)$. Since $i(x)\in I_{1}$ we have 
	\begin{eqnarray*}
	P^{0}(x_{I_{1}},x_{I_{2}},y|u_{I_{1}},u'_{I_{2}},v) & = & \overset{i(x)-1}{\underset{j=1}{\prod}}P(x_{j},y_{j}|u_{j},v_{j})\cdot P^{z_{i}=\sigma}(x_{i(x)},y_{i(x)}|u_{i(x)},v_{i(x)})\cdot\overset{n}{\underset{j=i(x)+1}{\prod}}P(x_{j},y_{j}|u'_{j},v_{j})\\
	 & = & \overset{i(x)-1}{\underset{j=1}{\prod}}P(x_{j},y_{j}|u_{j},v_{j})\cdot P^{z_{i}=\sigma}(x_{i(x)},y_{i(x)}|u_{i(x)},v_{i(x)})\cdot\overset{n}{\underset{j=i(x)+1}{\prod}}P(x'_{j},y_{j}|u{}_{j},v_{j})\\
	 & = & P^{0}(x_{I_{1}},x'_{I_{2}},y|u_{I_{1}},u{}_{I_{2}},v)
	\end{eqnarray*}
	and therefore Equation (\ref{eq:Alice-side-eq}) holds as well. 

	For the second case, assume that $i(x)\notin I_{1}$. $\forall x_{I_{1}},y,u_{I_{1}},u_{I_{2}},u'_{I_{2}},v$, denote by $u'=u_{I_{1}}u'_{I_{2}}$. Then 
	\begin{eqnarray*}
	\underset{x_{I_{2}}}{\sum}P^{0}(x_{I_{1}},x_{I_{2}},y|u_{I_{1}},u'_{I_{2}},v) & = & \underset{x_{I_{2}}}{\sum}\overset{i(x)-1}{\underset{j=1}{\prod}}P(x_{j},y_{j}|u'_{j},v_{j})\cdot P^{z_{i}=\sigma}(x_{i(x)},y_{i(x)}|u'_{i(x)},v_{i(x)})\cdot\overset{n}{\underset{j=i(x)+1}{\prod}}P(x_{j},y_{j}|u'{}_{j},v_{j})\\
	 & = & \underset{x_{I_{2}/i(x)}}{\sum}\overset{i(x)-1}{\underset{j=1}{\prod}}P(x_{j},y_{j}|u'_{j},v_{j})\cdot\\
	 &  & \cdot\left[P^{z_{i}=\sigma}(x_{i(x)},y_{i(x)}|u'_{i(x)},v_{i(x)})+P^{z_{i}=\sigma}(\overline{x_{i(x)}},y_{i(x)}|u'_{i(x)},v_{i(x)})\right]\cdot\\
	 &  & \cdot\overset{n}{\underset{j=i(x)+1}{\prod}}P(x_{j},y_{j}|u'{}_{j},v_{j})\\
	 & = & \underset{x_{I_{2}/i(x)}}{\sum}\overset{i(x)-1}{\underset{j=1}{\prod}}P(x_{j},y_{j}|u'_{j},v_{j})\cdot\\
	 &  & \cdot\left[P(x_{i(x)},y_{i(x)}|u'_{i(x)},v_{i(x)})+P(\overline{x_{i(x)}},y_{i(x)}|u'_{i(x)},v_{i(x)})\right]\cdot\\
	 &  & \cdot\overset{n}{\underset{j=i(x)+1}{\prod}}P(x_{j},y_{j}|u'{}_{j},v_{j})\\
	 & = & \underset{x_{I_{2}}}{\sum}\overset{n}{\underset{j=1}{\prod}}P(x_{j},y_{j}|u'_{j},v_{j})\\
	 & = & \underset{x_{I_{2}}}{\sum}P(x_{I_{1}},x_{I_{2}},y|u_{I_{1}},u'_{I_{2}},v)
	\end{eqnarray*}
	where the third equality is due to Equation (\ref{eq:property-1}). Now since the unbiased system $P$ fulfils all non-signalling conditions, and in particular it is also time-ordered non-signalling, we have $\underset{x_{I_{2}}}{\sum}P(x_{I_{1}},x_{I_{2}},y|u_{I_{1}},u'_{I_{2}},v)=\underset{x_{I_{2}}}{\sum}P(x_{I_{1}},x_{I_{2}},y|u_{I_{1}},u{}_{I_{2}},v)$. Adding everything together we get 
	\begin{eqnarray*}
	\underset{x_{I_{2}}}{\sum}P^{0}(x_{I_{1}},x_{I_{2}},y|u_{I_{1}},u'_{I_{2}},v) & = & \underset{x_{I_{2}}}{\sum}P(x_{I_{1}},x_{I_{2}},y|u_{I_{1}},u'_{I_{2}},v)\\
	 & = & \underset{x_{I_{2}}}{\sum}P(x_{I_{1}},x_{I_{2}},y|u_{I_{1}},u{}_{I_{2}},v)\\
	 & = & \underset{x_{I_{2}}}{\sum}P^{0}(x_{I_{1}},x_{I_{2}},y|u_{I_{1}},u{}_{I_{2}},v)\,.
	\end{eqnarray*}
	Therefor for both cases Equation (\ref{eq:Alice-side-eq}) holds and the system $P_{XY|UV}^{0}$ is time-ordered non-signalling. 
\end{proof}
The same proof holds for $P_{XY|UV}^{1}$ as well. 
\end{appendices}
\end{document}